\documentclass[preprint, authoryear, 11pt]{elsarticle}

\usepackage[titletoc,title]{appendix}
\usepackage{natbib}
\usepackage{amsfonts, amsmath}
\usepackage{bm}
\usepackage{setspace}
\usepackage{fullpage}
\usepackage{enumerate}
\usepackage{amsmath}%
\usepackage{amsfonts}%
\usepackage{amssymb}%
\usepackage{booktabs}
\usepackage{natbib}

\newcommand{\E}{\mathop{\mathbb E}}

\newcommand{\n}{\normalsize}
\newcommand{\s}{\scriptsize}

\newtheorem{theorem}{Theorem}

\newtheorem{corollary}[theorem]{Corollary}

\newtheorem{example}{Example}

\newtheorem{lemma}[theorem]{Lemma}

\newtheorem{proposition}[theorem]{Proposition}
\newtheorem{remark}{Remark}

\newproof{proof}{Proof}

\begin{document}
\title{\textsc{Strategy-Proof Incentives for Predictions\tnoteref{wine}}}
\tnotetext[wine]{A short version of this paper was presented in WINE'18 and published in proceedings}

\author{Amir Ban}
\ead{amirban@netvision.net.il}

\address{Faculty of Mathematics and Computer Science,
Weizmann Institute of Science,
Rehovot, Israel}

\begin{abstract}
Our aim is to design mechanisms that motivate all agents to reveal their predictions truthfully and promptly. For myopic agents, proper scoring rules induce truthfulness. However, as has been described in the literature, when agents take into account long-term effects of their actions, deception and reticence may appear. No simple rules exist to distinguish between the truthful and the untruthful situations, and a determination has been done in isolated cases only.

This is of relevance to prediction markets, where the market value is a common prediction, and more generally in informal public prediction forums, such as stock-market estimates by analysts.
We describe three mechanisms that are strategy-proof with non-myopic considerations, and show that one of them meets all our requirements from a mechanism in almost all prediction settings. We formulate rules to distinguish truthful from untruthful settings, and use them to extensively classify prediction settings with continuous outcomes. We show how our proposed mechanism restores prompt truthfulness where incumbent mechanisms fail, and offer guidelines to implementing it in a prediction market.

\end{abstract}

\maketitle

\section{Introduction}

Mechanisms that motivate all agents to reveal their information truthfully and promptly are desirable in many situations.

Consider, for example, the estimation of company earnings by stock-market analysts, a longstanding Wall Street institution. Publicly-traded companies announce their earnings for the latest quarter or year, on dates set well in advance. Each company is typically covered by several stock-market analysts, the larger ones by dozens. These analysts issue reports containing predictions of a company's future earnings. The timing of such predictions may range from several years to days before earnings announcement, and every analyst typically updates his prediction several times in the interval. These predictions eventually become publicly available, 
 and a consensus calculated from all predictions in force may be viewed on several popular finance websites. Not least, the analysts themselves are aware of, and are no doubt influenced by the actions and opinions of their peers.

In essence, this earnings estimation functions as a public prediction forum with an evolving consensus, that terminates when a company announces its true earnings for the forecast period, which we call the {\em outcome}. It acts as a sort of advisory forum for the public of investors, and this public's interest is best served if analysts share their information and judgement truthfully and promptly.

Prediction markets are public prediction forums organized as markets. A de-facto standard for organizing prediction markets is due to \cite{Hanson03}, using a {\em market scoring rule}. In such a market, probability estimates are rewarded by a {\em proper scoring rule} an amount $S(\bm{p},r)$, where $\bm{p}$ is a probability distribution of the outcome, and $r$ is the outcome. A trader in Hanson's markets not only makes her probability estimate public, she changes the market price to it. She then stands to be rewarded by the {\em market maker} for her prediction (when the outcome becomes known), but she also commits to compensate the {\em previous} trader for his prediction. Her total compensation is therefore the difference $S(\bm{p},r) - S(\bm{p'},r)$ where $\bm{p'}$ is the replaced market prediction. When the logarithmic scoring rule ($S(\bm{p},r) = \log p_r$)\footnote{We use the notation $p_x$ for the density of distribution $\bm{p}$ at $x$.} is used in a prediction market, the mechanism is called LMSR (Logarithmic Market Scoring Rule). Hanson also demonstrated how a market maker can facilitate, at a bounded cost to himself, such an LMSR market and provide liquidity by selling and buying shares of each outcome.

Proper scoring rules, are, by their definition, incentive compatible for {\em myopic} agents\footnote{The incentive compatibility is also restricted to {\em risk-neutral} agents. \cite{offerman2009truth} demonstrate how proper scoring rules may be corrected for other risk attitudes.}. That is, an agent maximizes her expected score by announcing her true belief, provided longer-term effects of the prediction, if any, are ignored. The incremental, market variation of the scoring rule does not affect this incentive compatibility, because the previous agent's score does not depend on the current prediction. Furthermore, it is a straightforward generalization of \cite{Hanson03} to apply market scoring rules to multiple-choice, or continuous outcomes (such as in our earnings estimate forum).

When an agent is {\em not} myopic, and {\em does} take into account all consequences of her action, truthfulness will, in many cases, {\em not} be her optimal strategy, and incentive compatibility is lost. Such scenarios have been described in the literature, and our paper adds many further examples. As we will show, the damage to incentive compatibility caused by long-term strategic considerations is extensive, and it is {\em a priori} unclear whether any remedy is available.

Our aim is the design of mechanisms for rewarding predictions that are strategy-proof. We demonstrate the problem with proper scoring rules, as often leading to reticence or deception. We formulate criteria for determining which prediction settings are truthful, and apply these criteria for
a complete classification of the important class of prediction settings with normal and lognormal signals. We suggest three strategy-proof mechanisms, and identify one of them, {\em discounting}, as having all desirable properties. We prove the applicability and effectiveness of the discounting mechanism, and offer guidelines to implementing it.

\subsection{The Problem with Scoring Rules}
\label{theproblem}
A scoring rule $S: \Delta(\mathbb{R}) \times \mathbb{R} \mapsto \mathbb{R}$ scores a prediction $\bm{p}$, representing a probability distribution of the outcome, a value $S(\bm{p}, r)$ when the outcome is $r$. An agent whose belief of the outcome distribution is $\bm{q}$ has score expectation $S(\bm{p}, \bm{q}) := \E_{r \sim \bm{q}} S(\bm{p}, r)$ for prediction $\bm{p}$. A {\em proper} scoring rule is one for which $S(\bm{q},\bm{q}) \geq S(\bm{p},\bm{q})$ for every $\bm{p}, \bm{q} \in \Delta(\mathbb{R})$, so that predicting one's true belief has maximal score expectation. A {\em strictly proper} scoring rule is one where the inequality is tight only for $\bm{p} = \bm{q}$. The {\em logarithmic} scoring rule $S(\bm{p}, r) = \log p_r$, and the {\em quadratic} (a.k.a. {\em Brier}) scoring rule $S(\bm{p}, r) = 2 p_r - \bm{p} \cdot \bm{p} - 1$ are examples of strictly proper scoring rules. More background on scoring rules may be found, e.g., in \cite{Gneiting07}.

The following generic example illustrates the problem when non-myopic considerations apply.
\begin{example}
\label{problem}
The public wants to predict a variable, whose outcome is $x$. A market scoring rule compensates predictions.
Every signal of $x$ is, i.i.d., $x + \epsilon$ with probability 1/2, and $x - \epsilon$ with probability 1/2. The distributions of $x$ and $\epsilon$ are such that a single signal does not reveal $x$. There is an expert, who gets private signals. The public gets public signals.
\begin{itemize}
\item On Sunday, expert gets a signal.
\item On Monday, public gets a signal.
\item On Tuesday, expert gets another signal. 
\item On Wednesday, outcome $x$ is revealed.
\end{itemize}

Question: Should expert reveal his information truthfully on Sunday?

Answer: \textsc{No}. Whoever sees two different truthful signals is able to calculate the outcome $x = (x + \epsilon)/2 + (x - \epsilon)/2$ exactly.
For any distribution of $\epsilon$, and for the logarithmic and almost\footnote{The example is true for the logarithmic scoring rule because its scores are unbounded. Every scoring rule that values exact predictions over inexact ones sufficiently will do.} every other scoring rule, the expert should not tell the truth on Sunday. This prevents the $50\%$ probability that the market will know $x$ on Monday, preserving a $75\%$ probability that the expert can announce $x$ on Tuesday.
\end{example}

The canonical case, to which this example belongs, is ``Alice-Bob-Alice'', where Alice speaks before and after Bob's single speaking opportunity, both are awarded by a proper scoring rule for each prediction, and both maximize their total score. \cite{Chen2007} as well as \cite{chen2016informational} studied situations where several agents, each having private information, are given more than one opportunity to make a public prediction. The situations are reducible to the Alice-Bob-Alice game. The proper scoring rule assures that each will tell the truth on their last prediction, and the open question is whether Alice, when going first, will tell the truth, lie, or keep her silence. \cite{Chen2007} make the key observation that truthfulness is optimal if, in a different setup, namely, a single-prediction Alice-Bob game where Alice chooses whether to go first or second, she will always prefer going first. Building on that insight, \cite{chen2016informational} show that when the players' information is what they define as ``perfect informational substitutes'', they will predict truthfully and as early as allowed, when they are ``perfect informational complements'', they will predict truthfully and as {\em late} as allowed, while when players are neither substitutes nor complements, untruthfulness can and will occur. While this characterization is helpful, few concrete cases have been settled. The most significant of those was to show that when signals are independent conditional on the outcome, and the logarithmic scoring rule (LMSR) is used, the signals are informational substitutes, meaning that in such a case, Alice will reveal all her information truthfully in her first round.

A strategy-proof mechanism of the Alice-Bob-Alice setting easily generalizes to a strategy-proof mechanism for any number of experts and prediction order, because
whenever an expert (call her Alice) makes more than one prediction, one can roll together all experts making predictions between Alice's successive predictions into one expert (call him Bob), who shares their information\footnote{\cite{Chen2007} use the same construction to generalize from Alice-Bob-Alice to a finite-players game.}. This is formally proved in Theorem \ref{general}.

\subsection{Goals of the Mechanism}
\label{goals}

We seek a mechanism with several desirable traits.

\begin{enumerate}
\item \textsc{Truthfulness}: The mechanism should motivate all experts to make truthful predictions, that is, to reveal their true subjective distributions of the outcome. At minimum, this means that truth-telling should be a best-response to truth-telling by all other experts, according to the player's beliefs at the time of prediction, i.e., it is a perfect Bayesian equilibrium.\footnote{Note that the ideal of truth-telling as dominant strategy is not attainable here, because if a player is aware of another player's distortion, the correct Bayesian response is to compensate for the distortion.}

An untruthful mechanism may still be {\em locally truthful} by which we mean that infinitesimal variations from the truth are suboptimal, but telling a sufficiently big lie may be advantageous. 

\item\textsc{Full Disclosure}: Every expert makes a (truthful) prediction some time after getting his last signal. This is a necessary\footnote{But possibly not sufficient, e.g. when experts predict the XOR of their bit signals.} condition for all the information on the outcome possessed in the expert's signals to reach the public.

\item \textsc{Promptness}: Experts should reveal their signals by a truthful public prediction as soon as the prediction schedule allows, and make an updated prediction whenever receiving a new signal, again, at the earliest opportunity. Considerations by experts of when to make a prediction are contrary to the interest of the public. We shall require a {\em strict} preference for promptness. Indifference to timing shall not count as prompt.

\item \textsc{Bounded Loss}: \cite{Hanson03} notes that in his market scoring rule mechanism, the market maker effectively subsidizes traders to motivate their truthfulness, and he shows that market maker's expected loss due to that is {\em bounded}: Since the market maker initializes the prediction (to some prior $\bm{\pi}$), and needs to compensate only the last prediction (as every trader compensates the previous one), his cost (=loss) expectation is, for non-positive scoring rules (such as LMSR)
\begin{align}
\label{loss}
\E[Loss] = S(\bm{p}_{final},\bm{\pi}) - S(\bm{\pi},\bm{\pi}) \leq  - S(\bm{\pi},\bm{\pi})
\end{align}

We seek mechanisms that have bounded loss. Note that scoring rules can be arbitrarily scaled without effect on their motivation, so reducing the  {\em size} of a bounded loss is not a goal.

\end{enumerate}

\subsection{Our Results}

We propose three different incentive mechanisms, all of which are based on proper scoring rules, all of which achieve truthfulness, and the third and last also achieves promptness. They are

\begin{enumerate}
\item \textsc{Group prediction}: All agents receive the final prediction's (non-incremental) score. Since all agents have a stake in the final prediction, all will reveal their information truthfully. On the negative side, they are not motivated to be prompt about it. Another problem is freeloaders, since agents with no information can participate and gain without contributing anything.

\item \textsc{Enforce single prediction}: Score each of agent's prediction with an incremental scoring rule, and award each agent the minimum score. Agents are therefore motivated to predict once only, since having made a prediction, a further prediction can only lower their reward expectation. With a proper scoring rule, this assures incentive compatibility with truthfulness. Agents are not motivated to be prompt, but instead need to find the optimal timing to make their single prediction. A major drawback is that when agents receive a time-varying signal, they will not reveal all their information.

\item \textsc{Discounting}: Discount each of agent's incremental prediction scores by a monotonically increasing factor of the time. The idea is that if signals are not informational substitutes, they will become ones if a sufficiently steep negative time gradient is applied. When successful, this mechanism achieves the ideal result of motivating all agents to reveal their information truthfully and promptly, including when they receive time-varying signals.

We show that, under some light conditions, discounting will always work unless signals are perfectly correlated, i.e., have a deterministic relation given the outcome (as in Example \ref{problem}).

\end{enumerate} 

Table \ref{mechanism-score-card} summarizes how the incumbent mechanism, and our three proposed mechanisms, measure up against each of the traits we described as desirable in the previous section.

We provide (Corollary \ref{min-truth}) a criterion for the truthfulness of the Alice-Bob-Alice game, as well as sufficient conditions for truthfulness and for untruthfulness (Corollary \ref{every-b}).

We put these to use to thoroughly investigate the Alice-Bob-Alice game (and, by extension, multi-player, multi-signal games) with both the logarithmic and the quadratic (Brier) scoring rule, when player signals have a 
multivariate normal distribution or a multivariate lognormal distribution. 

These distributions are among the best-known continuous distributions and naturally arise in many situations. 
They are characterized, {\em inter alia}, by the correlation coefficient ($\rho \in [-1,1]$) between Alice's and Bob's signals. When the scoring rule used is the logarithmic one, we find that when these signals are too well-correlated (whether positively or negatively), prompt truthfulness is {\em not} optimal. On the other hand, if the correlation is low, the game will be truthful and prompt. This includes the case $\rho = 0$, where, as is well-known, the signals are conditionally independent, confirming the \cite{Chen2007} result for conditionally independent signals.

However, when the quadratic scoring rule, one of the oldest and most commonly used scoring rules, is used with these multivariate distributions, it is {\em never} truthful for repeated predictions.

In all settings with either the logarithmic or the quadratic scoring rules, we show that our discounting mechanism restores prompt truthfulness, with the single exception of perfect correlation ($|\rho| = 1$) of the players' signals.


We show that the discounting mechanism can be effectively implemented with an automated market maker, as in Hanson's markets, thus showing that it may be practically applied in prediction markets. In the Discussion, we offer guidelines on how to design such a mechanism.

%
%
%

\subsection{Related Literature}

Scoring rules have a very long history, going back to \cite{Finetti37}, \cite{Brier50} and \cite{Good52}. Proper scoring rules are often used for incentive-compatible belief elicitation of risk-neutral agents (e.g. \cite{armantier2013eliciting}). Market scoring rules for prediction markets were introduced by \cite{Hanson03}.

The role of \cite{Chen2007} (which is based on earlier papers by \cite{chen2007bluffing} and \cite{dimitrov2008non}) and \cite{chen2016informational} in investigating the strategy-proofness of prediction markets was already described. \cite{gao2013you} and \cite{kong2018optimizing} resolve some more scenarios. \cite{conitzer2009prediction} embarks on a program similar to ours, citing mechanism design as a guiding principle. Accordingly, he strives to achieve the Revelation Principle, where all experts announce their private information to some organizing entity that makes the appropriate Bayesian aggregation. As we discuss in Section \ref{knowledge} below, we do not share that vision: Experts often do not know what part of their belief stems from truly private information, and even when they do, they cannot afford to go on record with a prediction which is not their best judgement. His ``Group-Rewarding Information Mechanism'' is similar to our Group Prediction mechanism, and its lack of fairness is pointed out. Conitzer does not propose a mechanism that achieves prompt truthfulness.

\cite{Chen2007} also suggest discounting, that ``reduces the opportunity for bluffing'', in their words, but does not prevent it (Section 9.1), so their discounting mechanism does not achieve our basic requirement of truthfulness. The reason is that their formulation is different from ours, applying same discount to {\em before} and {\em after} scores. On the other hand, we discount every prediction score according to the time its prediction was made. The difference is crucial, because theirs does not result in a true market scoring rule, as defined by \cite{Hanson03}. In consequence, our Section \ref{maximizing}, on which our results rest, as well as our Section \ref{market-maker}, do not apply to their formulation.

We shall occasionally rely on well-known facts of the normal and the multivariate normal distributions. The reader will find the basis for these in, e.g., \cite{Tong}.


The rest of this paper is organized as follows: In Section \ref{formulation} we formulate the problem. In Section \ref{already} we investigate which predictions settings are already truthful and prompt. In Section \ref{mechanisms} we offer strategy-proof mechanisms, and show how they can solve the gaps we have found. In Section \ref{discussion} we summarize and offer concluding remarks. Most proofs are to be found in the Appendix.

\section{Problem Formulation}
\label{formulation}

\subsection{Basics}

Two players, Alice and Bob, make public predictions of a real variable $\lambda$, whose prior distribution is $\bm{\pi}$. The {\em outcome} $x$ will be revealed after all predictions have been made. A prediction consists of revealing one's belief of the distribution of $\lambda$.  All agents (Alice, Bob, and the public) are Bayesian, and each prediction causes them to update their beliefs, i.e. the posterior distribution, of $\lambda$.\footnote{This formulation is different from the mechanism of prediction markets, but equivalent to it. In prediction markets, an agent replaces the current market prediction by his own. In our formulation, the agent merely announces a prediction, which, assuming the agent is truthful, becomes the market prediction by Bayesian inference. This is because all rational agents reach the same beliefs from the same data.}

In a truthful equilibrium, truthfulness is a best response of a player to truthfulness by all others. So we start by assuming that all agents take others' predictions as truthful, and investigate whether the best response is to be truthful.

The posterior beliefs are distributions of the variable $\lambda$ which are inferred from priors and likelihood functions using Bayesian techniques. In our discussion we find it more convenient and succinct to represent beliefs, without loss of generality, by a real number, rather than by a probability distribution. We use the fact that, in Bayesian analysis, when the likelihood functions belong to some family of distributions (e.g. exponential), all posterior beliefs belong to another family of distributions (Gamma distribution for the exponential family) $Q(\bm{Y}) \in \Delta(\lambda)$, called the {\em conjugate prior} of the first family. $\bm{Y}$ is a set of real {\em parameters} of the inferred distribution $Q$. We will assume models where one, and only one of these parameters is dependent on previous predictions, while the rest $\bm{Y} \setminus \{x\}$ is known from the model, the timing and the identity of the believer, but {\em does not depend on any previous prediction}. An example illustrates this:

\begin{example}
Assume Alice's belief of $\lambda$ to be normally distributed $N(\mu_A, 1 / \tau_A)$ where $\mu_A$ is the mean and $\tau_A$ the accuracy (i.e. inverse of the variance), and Bob's is $N(\mu_B, 1 / \tau_B)$ and independent of Alice's. $\tau_A$ and $\tau_B$ are set by the model and are commonly known. Assume an uninformative prior. Using well-known aggregation rules for independent normal observations, if Alice announces $\mu_A$, Bob's belief changes to the normally distributed $N(\mu_{AB},1 / \tau_{AB})$, where
\s\begin{align*}
\mu_{AB} &= \frac{\tau_A \mu_A + \tau_B \mu_B}{\tau_A + \tau_B} \\
\tau_{AB} &= \tau_A + \tau_B
\end{align*}\n

Notice that $\tau_{AB}$ can be calculated without knowing any of the means $\mu_A, \mu_B$, while $\mu_{AB}$ can be evaluated once $\mu_A$ and $\mu_B$ is known.

In this context we are therefore able to describe a prediction by a single real number (the mean) rather than by a probability distribution. We shall say that Alice's prior belief $A_1$ is $\mu_A$ and Bob's prior belief $B_1$ is $\mu_B$. After Alice makes her prediction, Bob's belief changes to $\mu_{AB}$. After Alice and Bob both make a prediction, the public's belief is $\mu_{AB}$. In context, these statements are unambiguously equivalent to specifying the probability distributions in full.
\end{example}

%
The prior $\bm{\pi}$ may be uninformative, assigning equal probabilities to all possibilities\footnote{Technically, an uninformative prior may be envisioned as the limit of a uniform or normal distribution as the variance goes to infinity.}, or, if not, as also representable by a parameter. For example, if Alice and Bob participate in a prediction market, the prior parameter is the market value before Alice's first prediction.

Time is discrete, $t = 0, 1, 2, \ldots, T$, with $T$ the time the outcome is known.  $A_t, B_t , C_t\in \mathbb{R}$ are, respectively, Alice's, Bob's and the public's (or market's) beliefs at time $t$. At $t=1$, $A_1, B_1, C_1$ are their respective prior beliefs. Any prediction takes place at $t > 1$. $t = 0$ is ``pre-prior'' time, when players beliefs are equal to their {\em private} signals, so that $A_0, B_0$ are respectively, Alice and Bob's private signals. At $t=1$, each player is additionally aware of the public prior $C_1$, so that $A_1$ is an inference from $A_0$ and $C_1$ while $B_1$ is an inference from $B_0$ and $C_1$. If the public prior is uninformative, then we have $A_0 = A_1$ and $B_0 = B_1$. In other words, the players' priors equal their private signals. For completeness, we define $C_0 = C_1$.

To avoid degenerate exceptions, we assume the players' signals are {\em informative}. This means that $A_1 \neq C_1$ and $B_1 \neq C_1$.

The signals have a common-knowledge joint distribution $f(a,b;\lambda)$ conditional on the variable
\s\begin{align*}
f(a,b;\lambda) := \Pr(A_1 = a, B_1 = b | \lambda)
\end{align*}\n

The order of predictions is Alice, Bob, then Alice again, and then the outcome $x$ is revealed.

A twice-differentiable, w.r.t. $\lambda$, proper scoring rule $S(\bm{p}, \lambda)$ incrementally rewards each prediction made, i.e., if a player's prediction changed the public's belief from $\bm{p'}$ to $\bm{p}$, the player's reward for this prediction, calculated when $x$ is known, is $S(\bm{p}, x) - S(\bm{p'},x)$. Each player seeks to maximize their total reward.

As the scoring rule is proper, Bob will tell the truth on his only prediction, and Alice will tell the truth on her second and last prediction. The remaining question is whether Alice will tell the truth on her first prediction. More accurately, the question is of equilibrium: If Bob is truthful, and Bob and the public take Alice's predictions as truthful, is truth-telling Alice's best response? If the answer is affirmative when all players use all available information to update their beliefs, we have a truthful Perfect Bayesian Equilibrium.

\subsection{Knowledge Model}
\label{knowledge}

As will be shown, the players behavior is affected by their {\em common-knowledge prior}, by which we mean a belief distribution which is explicitly known to both players and from which each inferred his or her current belief. That the common prior is commonly known is significant, because it is quite possible, and even likely, that the players share a prior but do not know what it is. For example, in predicting a poll, Alice and Bob may be basing themselves on knowledge of how their acquaintances voted, but they may not know which acquaintances they have in common. Or, they both may be basing themselves on a paper they read, but neither is aware that the other has read it. If the players do not know what their common prior is, they cannot infer anything from having one. 

In prediction markets, and more generally in public prediction forums where all communication is done in public, the common-knowledge prior is known. In the context of this paper, it is the initial market value $C_1$ (or, equivalently, distribution $\bm{\pi}$), when the Alice-Bob-Alice game starts. 

\subsection{Inferences from Predictions}
\label{inferences}

Assume that inference functions are invertible, so that if a player's prediction is known, her signal can be computed. If Alice announces $A_1 = a$, the posterior outcome distribution can be calculated from her marginal distribution, $f_A(a;\lambda) := \Pr(A_1 = a | \lambda) = \int_{-\infty}^{\infty} f(a, b';\lambda) db'$. Mark it $g(a)$.
\s\begin{align}
g(a)_\lambda = \Pr(\lambda | A_1 = a) = \frac{\pi_\lambda f_A(a;\lambda)}{\int_{-\infty}^{\infty} \pi_{\lambda'} f_A(a;\lambda') d\lambda'}
\end{align}\n

Similarly, if Alice announces $A_1 =  a$, and Bob privately observes $B_1 = b$, Bob's posterior outcome is inferred from $f(a, b; \lambda)$. It will become the public prediction when Bob announces it. Mark it $h(a,b)$.
\s\begin{align}
h(a, b)_\lambda = \Pr(\lambda | A_1 = a, B_1 = b) = \frac{\pi_\lambda f(a, b;\lambda)}{\int_{-\infty}^{\infty} \pi_{\lambda'} f(a, b;\lambda') d\lambda'}
\end{align}\n

\subsection{Maximizing the Reward}
\label{maximizing}

How does Alice maximize her total reward for both her predictions? And is this maximum achieved by telling the truth on both predictions? We will show that Alice maximizes her reward by {\em minimizing} Bob's reward, and therefore is truthful if truth minimizes Bob's reward.

\begin{theorem}
\label{minimize}
Alice maximizes her expected total reward by making a first prediction that minimizes Bob's expected reward, where expectations are taken according to Alice's beliefs on her first prediction, and assuming Bob is truthful and takes Alice's prediction as truthful.
\end{theorem}

\begin{proof}
See Appendix.
\end{proof}

\begin{corollary}
\label{min-truth}
If a truthful first prediction minimizes Alice's expectation of Bob's reward, i.e., if
\s\begin{align*}
\E\limits_{b \sim B_1|(A_1=a)} \E \limits_{\lambda \sim h(a,b)} \Pi_B(\lambda; a,\hat{a},b)  
= \Bigl\{\E\limits_{b \sim B_1|(A_1=a)} S(h(\hat{a},b),h(a,b))\Bigr\} - S(g(\hat{a}),g(a))
\end{align*}\n

\noindent is minimized at $\hat{a} = a$, then truth is Alice's best policy.
\end{corollary}

A corollary that we will find useful is the following.
\begin{corollary}
\label{every-b}
Let $c := \hat{a} - a$ be Alice's deviation from truth, and define \s$$\Delta(c;a, b) := \Bigl[S(h(\hat{a},b),h(a,b)) - S(h(a,b),h(a,b))\Bigr] - \Bigl[S(g(\hat{a}),g(a)) - S(g(a), g(a))\Bigr]$$\n

If, for every $c \neq 0$, and every $b$, $\Delta(c;a, b) > 0$, Alice will be promptly truthful.

Alternatively, if, for every $c \neq 0$, and every $b$, $\Delta(c;a, b) < 0$, Alice will not be promptly truthful.
\end{corollary}

\begin{proof}
If Alice assumes $B_1 = b$ when her belief is $g(a)$, then Alice's assumed belief is $h(a, b)$, and so her expectation of $\Pi_B(x; a,\hat{a},b)$ is (see \eqref{Pi_B} in Appendix), $S(h(\hat{a},b),h(a,b)) - S(g(\hat{a}),g(a))$. So the corollary assumes that truth ($\hat{a} = a$) minimizes Alice's expectation of Bob's reward for every $b$, and therefore also Alice's total expectation of Bob's reward. The corollary follows by Corollary \ref{min-truth}.
\qed
\end{proof}

\section{Which Prediction Settings are Already Truthful?}
\label{already}

We described above (Example \ref{problem}) an elementary setting that is not truthful, and a handful of other settings have been settled either way in the literature. But in the landscape of prediction settings that are of interest, the coverage has been very sparse. Beyond \cite{chen2016informational}'s criterion of ``Informational Substitutes'', which does not amount to an explicit algorithm\footnote{A submodularity property is required of the signal lattice, in a context described in their article.}, and \cite{kong2018optimizing}, who offer a relatively-efficient algorithm to settle some scenarios, we have no procedure to settle any given case, and the problem remains opaque.

With the results we derived in Section \ref{maximizing}, we now have such procedures. We shall apply them to classify settings belonging to the most commonly-met continuous distributions and the most commonly used scoring rules. The distributions are:
\begin{itemize}
\item The multivariate normal distribution $N(\bm{\mu},\bm{\Sigma})$. In contrast to general joint distributions, where the interdependence of components may take practically any form, the interdependence of joint multivariate normal components is completely determined by their covariance matrix. This means that in the Alice-Bob-Alice game, the signals interdependence is completely determined by a single real parameter: their correlation coefficient $-1 \leq \rho \leq 1$. Our investigation boils down to finding for what values of $\rho$ Alice will be promptly truthful.

%
\item The multivariate lognormal distribution $\bm{Y} = exp(\bm{X})$, where $\bm{X}$ is a multivariate normal distribution as described in the previous item. Stock prices, and their derivatives, are commonly modelled by a lognormal distribution. Since taking logs transform multivariate lognormal random variables to multivariate normal random variables, our classification of the multivariate normal distribution readily provides a classification of this distribution too.
\end{itemize}

The scoring rules we cover are
\begin{itemize}
\item The Logarithmic Scoring Rule $S(\bm{p},r) = \log p_r$. As mentioned, this proper scoring rule is employed as Hanson's LMSR market-maker. It has strong information theory roots, and is uniquely {\em local}, i.e. depending only on the distribution at the outcome ($r$).
\item The Brier/Quadratic Scoring Rule $S(\bm{p},r) = 2 p_r - \bm{p} \cdot \bm{p} - 1$. It is the earliest scoring rule, introduced by \cite{Brier50}, in the context of weather forecasting.\end{itemize}

\subsection{Multivariate Normal Distribution}

Let the public prior $\bm{\pi}$ be distributed $N(C_0, 1 / \tau_C)$, where $\tau_C = 1 / \sigma_C^2$ is the accuracy. $C_0$ is the public's prediction mean at the start of the game. Note that an uninformative prior is characterized by having an arbitrarily large variance $\sigma_C^2$, or equivalently $\tau_C = 0$, and in such a case, the value of $C_0$ is inconsequential.

Similarly, Alice's and Bob's priors are distributed $N(A_1, 1 / \tau_{AC})$ and $N(B_1, 1 / \tau_{BC})$, respectively, where $\tau_{AC} = 1 / \sigma_{AC}^2$ and $\tau_{BC} = 1 / \sigma_{BC}^2$ are the respective accuracies. $A_1, B_1$ are the means of, respectively, Alice's and Bob's prior predictions. Alice's prediction $A_1$ is inferred from the public prior $C_0$ and Alice's signal $A_0$, and similarly for Bob.

We lose no generality in assuming that each of the player's signals is conditionally independent of the prior. Because, if Alice, e.g., is aware of the public prior $C_0$ and her prior prediction $A_1$, she can consistently assume that her signal $\bm{\pi_A}$ is distributed $N(A_0, 1 / \tau_A)$, where
\s\begin{align}
A_0 &= \frac{\tau_{AC} A_1 - \tau_C C_0}{\tau_{AC} - \tau_C} &
\label{acc_1}
\tau_A &= \tau_{AC} - \tau_C
\end{align}\n

It is easy to verify that $\bm{\pi_A}$ and $\bm{\pi}$ are mutually independent. Note that $\tau_{AC} > \tau_C$ since we assume Alice's signal is informative, hence $\tau_A > 0$.

For an uninformative prior $\tau_C = 0$ this degenerates to $A_0 = A_1, B_0 = B_1$, as can be expected.

By assumption the private signals $\bm{\pi_A}, \bm{\pi_B}$ and the public prior $\bm{\pi}$ belong to a jointly multivariate normal distribution, and as such, their joint distribution is completely determined by their means and by their covariance matrix. Since $\bm{\pi_A}, \bm{\pi_B}$ are, w.l.o.g., conditionally independent of $\bm{\pi}$, the distributions are uncorrelated with $\bm{\pi}$, and the covariance matrix $M$ for $(\bm{\pi_A}, \bm{\pi_B}, \bm{\pi})$ is
\s\begin{align}
M = 
  \begin{pmatrix}
\label{covariance}
\sigma_A^2 & \rho \sigma_A \sigma_B & 0 \\
 \rho \sigma_A \sigma_B & \sigma_B^2 & 0 \\
 0 & 0 & \sigma_C^2
  \end{pmatrix}
\end{align}\n
where $\rho \in [-1,1]$ is the correlation coefficient of $\bm{\pi_A}$ and $\bm{\pi_B}$, defined as: \s$$\rho = \frac{Cov(\bm{\pi_A}, \bm{\pi_B})}{\sqrt{Var(\bm{\pi_A}) Var(\bm{\pi_B})}}$$\n As is well-known (e.g. \cite{Tong}), components of a multivariate normal distribution are independent iff they are uncorrelated, i.e. have $\rho = 0$.

%
%

\begin{proposition}
\label{abc}
Assume that the model, including covariance matrix $M$ \eqref{covariance}, is common knowledge. Then, knowing everyone's prior prediction $A_1, B_1$ and $C_0$, the posterior, marked $\bm{\pi_{ABC}}$, has distribution $N(\mu_{ABC},1 / \tau_{ABC})$, with
\s\begin{align}
\label{abc-mu}
\mu_{ABC} &= \frac{(\tau_A - \rho \sqrt{\tau_A \tau_B}) A_0 + (\tau_B - \rho \sqrt{\tau_A \tau_B}) B_0 + (1 - \rho^2) \tau_C C_0}{\tau_A - 2 \rho \sqrt{\tau_A \tau_B} + \tau_B + (1 - \rho^2) \tau_C} \\
\label{abc-tau}
\tau_{ABC} &= \frac{\tau_A - 2 \rho \sqrt{\tau_A \tau_B} + \tau_B}{1 - \rho^2} + \tau_C
\end{align}\n
\end{proposition}

\begin{proof}
See Appendix.
\end{proof}

\subsubsection{With Logarithmic Scoring}

\begin{proposition}
\label{lmsr-normal}
Let Alice's, Bob's and the public signals in an Alice-Bob-Alice game belong to a jointly multivariate normal distribution, with covariance matrix given in \eqref{covariance}. When using the logarithmic scoring rule, the players' best-response strategy to truthfulness in others is truthfulness and promptness iff
\s\begin{align}
\biggl(1-\rho^2\biggr)^2 \biggl(1 + \frac{\tau_C}{\tau_B}\biggr) \geq \biggl(\rho^2 + \frac{\tau_C}{\tau_A}\biggr)\biggl(\sqrt\frac{\tau_A}{\tau_B} - \rho\biggr)^2
\end{align}\n
\end{proposition}

\begin{proof}
See Appendix.
\end{proof}

In particular, we note in Proposition \ref{lmsr-normal} that the case $|\rho| = 1$ is untruthful, and more generally the same is true when the correlation coefficient is either too positive or too negative. The middle ground is truthful. In particular, substituting $\rho = 0$, where the players' signals are conditionally independent, is found to be truthful. This generalizes the result of \cite{Chen2007} (Theorem 5) for conditionally independent signals with the logarithmic scoring rule, because it shows that conditionally independent signals induce truthfulness even when modified by a common-knowledge prior.

Also in particular, if the common prior is uninformative, or unknown, we get the following corollary by substituting $\tau_C = 0$ in Proposition \ref{lmsr-normal} and solving for $\rho$.
\begin{corollary}
\label{uninformative}
Let Alice's and Bob's signals in an Alice-Bob-Alice game belong to a jointly multivariate normal distribution, with covariance matrix given in \eqref{covariance}, and an uninformative or unknown common prior ($\tau_C = 0$). Then, with the logarithmic scoring rule, the players' best-response strategy to truthfulness in others is truthfulness and promptness iff
\s\begin{align*}
\left\{  \begin{array}{ll}
\rho \leq \frac{\sigma_A}{\sigma_B} &$and$ \\
\frac{1}{4}\biggl(\frac{\sigma_B}{\sigma_A} - \sqrt{\frac{\sigma_B^2}{\sigma_A^2} + 8}\biggr) \leq \rho \leq \frac{1}{4}\biggl(\frac{\sigma_B}{\sigma_A} + \sqrt{\frac{\sigma_B^2}{\sigma_A^2} + 8}\biggr)
\end{array} \right.
\end{align*}\n
\end{corollary}

In particular, if $\sigma_A = \sigma_B$, truthfulness and promptness holds whenever $\rho \geq -\frac{1}{2}$.

\subsubsection{With Quadratic Scoring}

On the other hand, when using the Quadratic Scoring Rule, when the signals belong to a multivariate normal distribution, truthfulness is {\em never} optimal. This is summarized by the following proposition.
\begin{proposition}
\label{brier-normal}
Let Alice's, Bob's and the public signals in an Alice-Bob-Alice game belong to a jointly multivariate normal distribution, with covariance matrix given in \eqref{covariance}. Then, when using the quadratic scoring rule, truthfulness is {\em never} Alice's best-response strategy to truthfulness in others. Alice is, however, locally truthful for $0 < \rho < \frac{\sigma_A}{\sigma_B}$.
\end{proposition}

\begin{proof}
See Appendix.
\end{proof}

Remarkably, the result is unaffected by the prior ($\sigma_C$).

\section{Strategy-Proof Mechanisms}
\label{mechanisms}

\subsection{Group Prediction}

Our first strategy-proof mechanism scores the last prediction made by a proper scoring rule, and awards the score to {\em each} of the participating experts.

\begin{proposition}
\label{group}
Let $x$ be the outcome, and let the last prediction made before the outcome is revealed be distribution $\bm{p}$. Then the mechanism that awards $S(\bm{p},x)$ to each participating player, where $S$ is a proper scoring rule, is truthful. Furthermore, the mechanism elicits full disclosure.
\end{proposition}

\begin{proof}
See Appendix.
\end{proof}


The mechanism does not motivate promptness, which we defined as a strict preference: Players may predict as late as possible without harming their welfare. 

Another drawback is unfairness: The mechanism awards all experts the same, regardless of their contribution. Indeed, a so-called expert who has no information of his own may reap the same reward as other experts by simply repeating the current public prediction. This means that, unless the number of experts is bounded, the mechanism is not loss-bounded. Attempts to fix this would be counterproductive, compromising truthfulness. For example, rewarding nothing to ``predictions'' that merely repeat the public prediction, will motivate an uninformed expert to pretend knowledge by ``tweaking'' the current public prediction.

\begin{table}[tb]
  \caption{Mechanism Score Card}
  \centering
 \begin{tabular}{lcccc}
\toprule
Mechanism & Truthful & Full Disclosure & Prompt & Bounded Loss \\
\midrule\midrule
Market Scoring Rule & $\times$ & \checkmark & $\times$ & \checkmark \\
\midrule
Group Prediction & \checkmark & \checkmark & $\times$ & $\times$\\
Enforce Single Prediction & \checkmark & $\times$ & $\times$ & \checkmark \\
Discounting & \checkmark & \checkmark & \checkmark & \checkmark \\
\bottomrule
\end{tabular}
\label{mechanism-score-card}
\end{table}

\subsection{Enforce Single Prediction}

Since multiple predictions are the source of the potential for manipulation, a mechanism that prevents that would restore general truthfulness.

\begin{proposition}
Let $S$ be a proper scoring rule. The mechanism that scores each prediction with the increment $S(\bm{p}, x) - S(\bm{p'}, x)$, where $\bm{p}$ is the predicted distribution and $\bm{p'}$ the previous public prediction, and rewards each expert with the {\em minimum} score out of all her predictions, is truthful.
\end{proposition}

\begin{proof}
Once an expert has made a prediction, a further prediction may only lower her reward expectation and so is not optimal. Since incremental, proper scoring rules are truthful with a single prediction, the mechanism is truthful.
\qed
\end{proof}

While this mechanism is truthful, and loss-bounded (a consequence of Hanson's loss-boundedness. See Section \ref{goals} \eqref{loss}), it does not motivate promptness. Every expert 
needs to figure out the best time to make his single prediction, given other players' strategies, resulting in a Bayes-Nash equilibrium. As \cite{azar2016should} show, this can be complex. Furthermore, full disclosure is not assured, since experts may choose to make a prediction before getting their final signal. A further caveat is the assumption that expert identity can be verified by the mechanism.

\subsection{Discounting}
\label{discounting}

The last, and, we will argue, the most successful mechanism we suggest for incentive compatibility is {\em discounting}.

Discounting essentially uses a proper market scoring rule. As explained in the Introduction, for a proper scoring rule $S(\bm{p},r)$, the market scoring rule scores a prediction $\bm{p}$  $S(\bm{p},r) - S(\bm{p'},r)$ where $\bm{p'}$ is the outcome distribution that was replaced by $\bm{p}$, and $r$ is the outcome. Now any scoring rule can be scaled by an arbitrary constant $k$ and remain proper. Furthermore, $k S(\bm{p},r) - k' S(\bm{p'},r)$ where $k, k'$ may be different, is also proper. Generally, we can employ a time-varying scale factor $k(t) \in \mathbb{R}_{>0}$, and use the proper scoring rule $k(t) S(\bm{p}, r)$, where $t$ is the time (which may be the elapsed time from some base, or an integer counter of events) of announcement of $\bm{p}$. Discounting means choosing $k(t)$ that is weakly decreasing in $t$, whence we get a discounted scoring rule $D(\bm{p}, r, t) :=  k(t) S(\bm{p}, r)$, where $t$ is the time of prediction $\bm{p}$. When $t' < t$ is the time replaced prediction $\bm{p'}$ was made, the discounted market scoring rule for prediction $\bm{p}$ is\s$$D(\bm{p}, r, t) - D(\bm{p'},r, t') =  k(t) S(\bm{p}, r) - k(t') S(\bm{p'}, r)$$\n

The idea of discounting is that, if without discounting Alice's optimal strategy is to hide or distort information on her first prediction, in order to reap a bigger benefit on her second prediction, her calculus will change if a sufficiently steep discount, motivating earlier predictions, is imposed on her reward. 

We must be careful to use non-positive scoring rules with the discounting mechanism. Otherwise a possibility exists that the discounted scoring rule will have negative expectation for a prediction, creating a situation in which a player will prefer {\em not} predicting to making a truthful prediction. On the other hand, for non-positive scoring rules, $S(\bm{p}, r) \leq 0$, so that, whenever $0 \leq k(t) \leq k(t')$ \s$$k(t) S(\bm{p}, r) - k(t') S(\bm{p'}, r) \geq k'(t) [S(\bm{p}, r) - S(\bm{p'}, r)]$$\n So, if our original scoring rule had positive expectation, so does our discounted one.

Note that the logarithmic scoring rule $\log p_r$ and quadratic scoring rule $2 p_r - \bm{p} \cdot \bm{p} - 1$, which we have analyzed, are non-positive. Any scoring rule can be made non-positive by affine transform.

Being a market scoring rule, the discounting mechanism has bounded loss. More specifically

\begin{proposition}
\label{loss-bound}
A non-negative discounted market scoring rule $k(t) S(\bm{p}, r)$, where $t$ is the time of prediction $\bm{p}$ and $k(t)$ is a discounting function, has loss expectation bounded by $-k(0) S(\bm{\pi}, \bm{\pi})$, where $\bm{\pi}$ is the prior ($t=0$ prediction).
\end{proposition}

\begin{proof}
The expectation of $k(t) S(\bm{p}, r)$ for $r \sim \bm{p}$ is $k(t) S(\bm{p}, \bm{p})$. Since the discounted scoring rule is a market scoring rule per \cite{Hanson03}, its expected loss for $r \sim \bm{\pi}$ and $t = 0$ is, by \eqref{loss} (Section \ref{goals}), at most $-k(0) S(\bm{\pi}, \bm{\pi})$.
\end{proof}

\begin{remark}
Non-positive scoring rules suffer from an artefact that is the mirror image of positive ones: Experts have positive score expectation ($[k(t) - k(t')] S(p,r)$) for merely repeating the current prediction, and so will do so even if they have no information. In fact, the mechanism can offer them this ``reward'' automatically, sparing them the need to make an empty prediction. While this is ugly, its effect is minor, and the expected loss is still bounded by $-k(0) S(\bm{\pi}, \bm{\pi})$ by Proposition \ref{loss-bound}. As in other cases, attempts to mend this would be counterproductive. E.g., if we deduct the ``unearned'' $[k(t) - k(t')] S(p,r)$ from every score, we will compromise truthfulness, as this reverts to the discounting mechanism proposed in \cite{Chen2007}.
\end{remark}

\begin{proposition}
\label{existK}
Using the notation of Section \ref{inferences}, if $S$ is a non-positive scoring rule, and there exists $K$ such that
\s\begin{align}
\label{Kfactor}
K \geq \frac{\E\limits_{b \sim B_1|(A_1=a)}\bigl[S(h(a,b),h(a,b) - S(h(\hat{a},b),h(a,b))\bigr]}{S(g(a),g(a)) - S(g(\hat{a}),g(a))}
\end{align}\n

\noindent for every $A_1 = a$ and every possible $\hat{a}$, then the game is promptly truthful with a discount factor $k(t)$ that satisfies $k(t_1) / k(t_2) \geq K$, where $t_1$ is the time of Alice's first prediction and $t_2$ is the time of the second.
\end{proposition}

\begin{proof}
If \eqref{Kfactor} is satisfied, scoring rule $k(t) S(p, r)$ satisfies Corollary \ref{min-truth}.
\qed
\end{proof}

From Proposition \ref{existK}, we state sufficient conditions for discounting to succeed.
\begin{corollary}
\label{discount-crit}
If
\begin{enumerate}
\item $S$ is a non-positive strictly proper scoring rule, and
\item $\E\limits_{b \sim B_1|(A_1=a)}\bigl[S(h(a,b),h(a,b) - S(h(\hat{a},b),h(a,b))\bigr]$ is bounded for every $a, \hat{a}$
\begin{itemize}
\item In particular, if $S$ is the quadratic scoring rule, or any other bounded scoring rule, and
\end{itemize}
\item The right-hand side of \eqref{Kfactor} is bounded when $\hat{a} \to \pm\infty$, and
\item $lim_{\hat{a} \to a} \frac{ \frac{\partial^2}{\partial \hat{a}^2} S(h(\hat{a},b),h(a,b))}{ \frac{\partial^2}{\partial \hat{a}^2} S(g(\hat{a}),g(a))}$ exists for every $a, b$.
\end{enumerate}

Then there exists a discount factor effective at restoring truthfulness.
\end{corollary}

\begin{proof}
As $S$ is strictly proper, and the nominator of \eqref{Kfactor} is bounded, \eqref{Kfactor} can be unbounded only at infinity or at $\hat{a} = a$. At $\hat{a} = a$, \eqref{Kfactor} evaluates to $\frac{0}{0}$, so we invoke L'H\^{o}pital's rule to find the limit. By the definition of a proper score rule, first derivatives again evaluate to $\frac{0}{0}$, so invoke L'H\^{o}pital again for second derivatives.
\qed
\end{proof}

Discounting can restore truthfulness in the untruthful settings we found in Section \ref{already}.

\begin{proposition}
\label{lmsr-normal-discount}
Let Alice's, Bob's and the public signals in an Alice-Bob-Alice game belong to a jointly multivariate normal distribution, with covariance matrix given in \eqref{covariance}. When using the discounted logarithmic scoring rule $k(t) \log p_r$, players are truthful and prompt iff signals are not perfectly correlated (i.e. $|\rho| < 1$) and
\s\begin{align}
\label{discount}
k(t_1) / k(t_2) \geq  \frac{1 + \tau_C / \tau_A}{1 - \rho^2} \Biggl/ \Biggl[1 + \frac{(1-\rho^2)(\tau_B + \tau_C)}{\bigl(\sqrt{\tau_A} - \rho \sqrt{\tau_B}\bigr)^2}\Biggr]
\end{align}\n

\noindent where $t_1$ is the time of Alice's first prediction and $t_2$ is the time of the second. 
\end{proposition}

\begin{proof}
See Appendix.
\end{proof}

Since Alice's signal is informative, $\tau_A > 0$, and Proposition \ref{lmsr-normal-discount} shows that discounting will be effective unless $|\rho| = 1$. Note that when the right-hand side of \eqref{discount} evaluates to $\leq 1$, discounting is unnecessary. E.g. for $\rho = 0$, we have determined (Proposition \ref{lmsr-normal}) that truthfulness exists, and \eqref{discount} evaluates to $k(t_1) / k(t_2) \geq \frac{\tau_A + \tau_C}{\tau_A + \tau_B + \tau_C}$. Thus setting $k(t_1) / k(t_2) = 1$ is satisfactory.


\begin{proposition}
\label{brier-normal-discount}
Let Alice's, Bob's and the public signals in an Alice-Bob-Alice game belong to a jointly multivariate normal distribution, with covariance matrix given in \eqref{covariance}. Then, iff signals are not perfectly correlated (i.e. $|\rho| < 1$), there exists a discounted quadratic scoring rule $k(t) [2 p_r - \bm{p} \cdot \bm{p} - 1]$ for which players are truthful and prompt.
\end{proposition}

\begin{proof}
See Appendix.
\end{proof}

Finally, we can formulate a proposition for the prompt truthfulness of a general prediction forum.

\begin{theorem}
\label{general}
The public and a set of experts predict a variable $\lambda$. Let there be a fixed schedule in $[0,T]$ specifying when experts and public receive signals, and experts may make a prediction. Agents may receive multiple signals, and experts may have multiple prediction opportunities.

Then the forum is generally truthful and prompt with a discount function $k(t) \in \mathbb{R}_{>0}$ that makes all Alice-Bob-Alice {\em subgames} in the schedule truthful and prompt. These subgames are all occurrences where any expert, identified as ``Alice'', has two consecutive prediction opportunities, and all experts who makes predictions in between are rolled into a single player identified as ``Bob''.
\end{theorem}

\begin{proof}
See Appendix.
\end{proof}

%
%

\section{Discussion}
\label{discussion}

\subsection{Conclusions}

We showed that using proper scoring rules for rewarding predictions often leads to reticence or deception.
We formulated criteria for determining which prediction settings are truthful,
and made a complete classification of the class of prediction settings with normal and lognormal signals, under the logarithmic and quadratic scoring rules.
We suggested three new strategy-proof mechanisms, and identified one of them, discounting, as having all desirable properties,
and proved the applicability and effectiveness of the discounting mechanism. Guidelines to implementing it are offered in Section \ref{notes} below.

\subsection{A Market Maker for the Discounting Mechanism}
\label{market-maker}

\cite{Hanson03} has shown that every market scoring rule can be implemented by an automated market maker, who provides liquidity, and is willing to make any fair trade in shares of the form ``pay \$1 if outcome is $r$'', for every $r$. This applies to our discounting mechanism for the logarithmic scoring rule, which, as noted, takes the form of a market scoring rule for $S(\bm{p},r) = k(t) \log p_r$.

Applying Hanson's explanations to our case, if at $t$ there is an inventory $\bm{s}$ of shares $s_r$ for every $r$, the instantaneous share price is \s$$m(s_r) = e^{\frac{s_r}{k(t)}} \big/ \int_{-\infty}^{\infty} e^{\frac{s_x}{k(t)}} dx$$\n
Assuming an infinitesimal trade path, this induces a cost function $C(\bm{s}, t)$. The cost of a trade changing the inventory from $\bm{s}$ to $\bm{s'}$ at $t'$, is \s$$C(\bm{s'}, t') - C(\bm{s}, t) := k(t') \log \int_{-\infty}^{\infty} e^{\frac{s'_x}{k(t')}} dx - k(t) \log \int_{-\infty}^{\infty} e^{\frac{s_x}{k(t)}} dx$$\n

Since $k(t)$ is weakly decreasing in $t$, differentiating the above expression with respect to $t'$ shows that the more any trade with the automated market maker is delayed, the higher its cost.

\subsection{Implementation Notes}
\label{notes}

To implement the discounting mechanism $k(t) S(p, r)$, implement a standard LMSR prediction market with $S(p,r) = \log p_r$, adding a discounting scheme $k(t)$.

We offer some advice on how to select the discounting scheme, which is admittedly speculative or vague at times. Let us suppose that our prediction market is organized for the example which opens the paper: the estimation of company earnings by analysts. Instead of (or in addition to) issuing research reports, the analysts make predictions in the prediction market.

Proposition \ref{existK}, in conjunction with Theorem \ref{general}, specifies a lower bound on the discount factor to be effective. The bound is dependent on the information model, which is typically known only in outline, and often only from empirical data. There is no upper bound, so imposing a discount factor so steep that it surely exceeds the required minimum is possible. This rough and simple method would work, technically, because the motivation provided by scoring rules is unaffected by scale. Practically, however, when translated into actual compensation, size does matter. To design more subtle, but still suitable discounting, we point out when discounting is most needed, and when it is not.

As is apparent from Example \ref{problem} and throughout the paper, withholding or distorting information makes sense only until that information becomes known anyway. Regardless of discounting, Alice will not distort information that Bob is likely to find out through his own devices. In most multi-player settings, like the analyst forum, private data and insights are temporary: What Alice knows now, Bob will figure out sooner or later, and vice versa. Deception is time-limited.

It is therefore probably safe to {\em not} discount during periods of prediction inactivity, or infrequent activity. It may even be safe to {\em reset} the discount factor to its original value at a quiet prediction period: Even if Alice knows about the reset in advance, and calculates that it is profitable to lie about an insight or information she has, her scheme will be foiled if her insight, or information, will become known to others through other means. If the reset is too far in the future, this becomes a losing bet. And, so long as the number of resets is bounded, the mechanism is still loss-bounded. 

Of course, an {\em unannounced} reset of the discount factor is even better, as distortion in anticipation of a reset is not possible.

Discounting is most important when an information {\em shock} occurs. This is an event, or news, that reaches all analysts, and creates a scramble to update their previous now-obsolete estimate. An example is when the company issues an earnings warning, or new competition is announced. When this happens, all analysts suddenly have highly-correlated private information (by definition, all information not encompassed in the market price). Until a new market price stabilizes, there is great scope for deception. For example, the first analyst might try to deceive his peers about how severely he thinks the news will affect the earnings number. A strong discount is needed to put aside such thoughts.

It therefore seems appropriate to discount based on prediction {\em number} rather than prediction {\em time}, i.e. $k(t)$ where $t$ is a prediction counter rather than time-of-day, e.g., $k(t) = a^t$ where $a \in (0,1)$. This has the merit of discounting heavily in periods of high activity, and lightly in periods of low activity. Resetting the discounting, by resetting $t$ to $0$, at unannounced random times may supplement this, to keep the discount factor from falling too low. 

\section*{References}
\bibliographystyle{ACM-Reference-Format-Journals}
\bibliography{infomarkets}

\newpage
\appendix

\section{Proof of Theorem \ref{minimize}}

\begin{proof}
Let $A_1 = a, B_1 = b$. Assume that Alice, on her first prediction, announces her prediction as $\hat{a}$, where possibly $\hat{a} \neq a$. Alice's belief at this time is $g(a)$, while all other agents' beliefs after Alice's announcement is $g(\hat{a})$. As discussed, Bob's prediction and Alice's second prediction are truthful. Therefore, after Bob's prediction the public belief is $h(\hat{a},b)$. After Alice's second and last prediction, which is truthful, the public belief is $h(a,b)$ .

Alice's total reward given outcome $x$, $\Pi_A(x)$ is
\s\begin{align}
\Pi_A(x; a,\hat{a},b) = \Bigl\{S(g(\hat{a}),x) - S(\bm{\pi},x)\Bigr\} + \Bigl\{S(h(a,b),x) - S(h(\hat{a},b),x)\Bigr\}
\end{align}\n

Bob's total reward given outcome $x$, $\Pi_B(x)$ is
\s\begin{align}
\label{Pi_B}
\Pi_B(x; a,\hat{a},b) = S(h(\hat{a},b),x) - S(g(\hat{a}),x)
\end{align}\n

The sum of Bob's and Alice's rewards is $\Pi_A(x; a,\hat{a},b) + \Pi_B(x; a,\hat{a},b) = S(h(a,b),x) - S(\bm{\pi},x)$, which does not depend on $\hat{a}$. Therefore the $\hat{a}$ that maximizes Alice's reward is the one that minimizes Bob's reward. This is true for every $x$ and $b$, and therefore it is also true for Alice's reward expectation at the time of her first prediction. This is before Alice learned that $B_1 = b$, but had a conditional distribution of it $b \sim B_1|(A_1 = a)$ based on her belief. In summary

\s$$\arg\max_{\hat{a}} \E\limits_{b \sim B_1|(A_1=a)} \E \limits_{\lambda \sim h(a,b)} \Pi_A(\lambda; a,\hat{a},b) = \arg\min_{\hat{a}} \E\limits_{b \sim B_1|(A_1=a)} \E \limits_{\lambda \sim h(a,b)}  \Pi_B(\lambda; a,\hat{a},b)$$\n
\qed
\end{proof}

%
%
%

\section{Proof of Proposition \ref{abc}}

\begin{proof}
We know the player signals $A_0, B_0$ from their prior predictions $A_1, B_1$ by \eqref{acc_1}.

First, let us calculate the posterior distribution $\bm{\pi_{AB}} \sim N(\mu_{AB}, 1 / \tau_{AB})$ from Alice's signal $\bm{\pi_A} \sim N(A_0, 1 / \tau_A)$ and Bob's, $\bm{\pi_B} \sim N(B_0, 1 / \tau_B)$, alone.

Define an auxiliary signal $\bm{\pi_B'} = (1-\alpha) \bm{\pi_A} + \alpha \bm{\pi_B}$ where $\alpha := \frac{1}{1 - \rho\sqrt{\frac{\tau_A}{\tau_B}}}$, whose mean is $B_0' = (1 - \alpha) A_0 + \alpha B_0$. Its variance is
\s\begin{align*}
\sigma_B'^2 &= (1 - \alpha)^2 Var  \bm{\pi_A} + 2 \alpha (1 - \alpha) cov ( \bm{\pi_A},  \bm{\pi_B}) + \alpha^2 Var  \bm{\pi_B} \\
&= (1 - \alpha)^2 \sigma_A^2 + 2 \alpha (1 - \alpha) \rho \sigma_A \sigma_B + \alpha^2 \sigma_B^2
\end{align*}\n
\noindent and accuracy $\tau_B' = 1 / \sigma_B'^2$.

$\bm{\pi_A}, \bm{\pi_B'}$ are mutually independent, because
\s\begin{align*}
Cov(\bm{\pi_A}, \bm{\pi_B'}) = (1 - \alpha) Var \bm{\pi_A} + \alpha Cov(\bm{\pi_A}, \bm{\pi_B}) = \frac{1 - \alpha}{\tau_A} + \frac{\alpha \rho}{\sqrt{\tau_A \tau_B}} = \frac{-\frac{\rho}{\sqrt{\tau_A \tau_B}} + \frac{\rho}{\sqrt{\tau_A \tau_B}}}{1 - \rho \sqrt\frac{\tau_A}{\tau_B}} = 0
\end{align*}\n

Then $\bm{\pi_{AB}}$ is inferred from $\bm{\pi_A}, \bm{\pi_B'}$ by the standard inference rule for independent normal observations.
\s\begin{align*}
\tau_{AB} &= \tau_A + \tau_B' \\
\mu_{AB} &= \frac{\tau_A A_0 + \tau_B' B_0'}{\tau_{AB}}
\end{align*}\n

Substituting for $\alpha$ and simplifying, we get
\s\begin{align}
\label{muAB}
\mu_{AB} &= \frac{(\tau_A - \rho \sqrt{\tau_A \tau_B})A_0 + (\tau_B - \rho \sqrt{\tau_A \tau_B})B_0}{\tau_A - 2 \rho \sqrt{\tau_A \tau_B} + \tau_B}\\
\label{tauAB}
\tau_{AB} &= \frac{\tau_A - 2 \rho \sqrt{\tau_A \tau_B} + \tau_B}{1 - \rho^2}
\end{align}\n

Now, as the public signal $\bm{\pi} \sim N(\bm{\pi}, 1 / \tau_C)$ is independent of both $\bm{\pi_A}$ and $\bm{\pi_B}$, it is independent of $\bm{\pi_{AB}}$. Applying again the inference rule for independent normal observations, we have
\s\begin{align}
\label{tauABC}
\tau_{ABC} &= \tau_{AB} + \tau_C \\
\label{muABC}
\mu_{ABC} &= \frac{\tau_{AB} \mu_{AB} + \tau_C C_0}{\tau_{AB} + \tau_C}
\end{align}\n

Substituting \eqref{muAB} and \eqref{tauAB} in \eqref{muABC} and \eqref{tauABC} we get \eqref{abc-mu} and \eqref{abc-tau}.
\qed
\end{proof}

\section{Proof of Proposition \ref{lmsr-normal}}

We first prove a lemma for divergences of the logarithmic scoring rule.
\begin{lemma}
\label{log-diverge}
Let $s$ be the logarithmic scoring rule, and let $\bm{p} \sim N(\mu, 1/\tau)$ and $\bm{\hat{p}} \sim N(\hat{\mu}, 1/\tau)$ be two normal distributions with the same accuracy $\tau$. Then
\s\begin{align}
S(\bm{\hat{p}},\bm{p}) - S(\bm{p},\bm{p}) = -\frac{\tau}{2}(\hat{\mu} - \mu)^2
\end{align}\n
\end{lemma}

\begin{proof}

\s\begin{align*}
S(\bm{\hat{p}},\bm{p}) &= \frac{\tau}{\sqrt{2\pi}} \int\limits_{-\infty}^\infty e^{-\frac{\tau}{2}(x - \mu)^2} \log\Bigl\{(\frac{\tau}{\sqrt{2\pi}} e^{-\frac{\tau}{2}(x - \mu)^2}\Bigr\} dx \\
&= \log \frac{\tau}{\sqrt{2\pi}} - \frac{\tau}{\sqrt{2\pi}} \int\limits_{-\infty}^\infty e^{-\frac{\tau}{2}(x - \mu)^2} \Bigl\{\frac{\tau}{2}(x - \hat{\mu})^2\Bigr\}dx \\
&=  \log \frac{\tau}{\sqrt{2\pi}} - \frac{\tau}{2}\Bigl[\frac{1}{\tau} + (\hat{\mu} - \mu)^2\Bigr] \\
&=  \log \frac{\tau}{\sqrt{2\pi}} - \frac{1}{2} - \frac{\tau}{2}(\hat{\mu} - \mu)^2
\end{align*}\n

\noindent where we used the fact that the second moment of $N(\mu, \sigma^2)$ is $\sigma^2 + \mu^2$.

Consequently (substitute $\bm{p}$ for $\bm{\hat{p}}$ above), \s$$S(\bm{p},\bm{p}) =  \log \frac{\tau}{\sqrt{2\pi}} - \frac{1}{2}$$\n and the lemma follows.
\qed
\end{proof}

The proof of the proposition follows.

\begin{proof}

We use Corollary \ref{every-b}. Suppose that $A_1 = a$, and that Alice, when making her first prediction, pretends that $A_1 = \hat{a} := a + c'$, where $c' = \frac{\tau_A}{\tau_A + \tau_C}c$. I.e., by \eqref{acc_1}, Alice pretends her signal is $A_0 + c$.

 By the Corollary, Alice is truthful if \s$$\Delta(c;a, b) := \Bigl[S(h(\hat{a},b),h(a,b)) - S(h(a,b),h(a,b))\Bigr] - \Bigl[S(g(\hat{a}),g(a)) - S(g(a), g(a))\Bigr] \geq 0$$\n for every $\hat{a}$ and $b$. $g(a)$, Alice's distribution at her first prediction, is given by \eqref{acc_1}. When $B_1 = b$, the inferred distribution from Alice's and Bob's predictions $h(a, b)$ is given by Proposition \ref{abc}. Substituting these, and using Lemma \ref{log-diverge}:
\s\begin{align}
\Delta(c; a, b) &= \frac{\tau_{AC}}{2} \Bigl(c\frac{\tau_A}{\tau_A + \tau_C}\Bigr)^2 - \frac{\tau_{ABC}}{2}\Bigl(c \frac{(\tau_A - \rho \sqrt{\tau_A \tau_B})}{{\tau_A - 2 \rho \sqrt{\tau_A \tau_B} + \tau_B + (1 - \rho^2) \tau_C}}\Bigr)^2 \nonumber \\
&= c^2 \Bigl\{ \frac{\tau_A^2}{\tau_A + \tau_C} - \frac{(\tau_A - \rho \sqrt{\tau_A \tau_B})^2}{(1 - \rho^2)(\tau_A - 2 \rho \sqrt{\tau_A \tau_B} + \tau_B + (1 - \rho^2) \tau_C)}\Bigr\} \nonumber \\
\label{c^2}
&= c^2 \tau_A \Bigl\{ \frac{\tau_A}{\tau_A + \tau_C} - \frac{(\sqrt{\tau_A} - \rho \sqrt{\tau_B})^2}{(1 - \rho^2)(\sqrt{\tau_A} - \rho \sqrt{\tau_B})^2 + (1 - \rho^2)^2 (\tau_B + \tau_C))}\Bigr\}
\end{align}\n

\eqref{c^2} shows that $c=0$, i.e., truthfulness, is either a global minimum or a global maximum of $\Delta(c; a,b)$ depending on the sign of the factor multiplying $c^2$. Furthermore, this factor does not depend on $b$ (nor on $a$), and therefore Corollary \ref{every-b} applies: Alice is promptly truthful iff
\s\begin{align}
\label{crit}
\frac{\tau_A}{\tau_A + \tau_C} - \frac{(\sqrt{\tau_A} - \rho \sqrt{\tau_B})^2}{(1 - \rho^2)(\sqrt{\tau_A} - \rho \sqrt{\tau_B})^2 + (1 - \rho^2)^2 (\tau_B + \tau_C))} \geq 0
 \end{align}\n

Rearranging, we get
\s\begin{align}
1 - \rho^2 + (1 - \rho^2)^2\frac{\tau_B + \tau_C}{(\sqrt{\tau_A} - \rho \sqrt{\tau_B})^2} \geq \frac{\tau_A + \tau_C}{\tau_A} \Rightarrow \\
(1 - \rho^2)^2 (\tau_B + \tau_C) \geq (\rho^2 + \frac{\tau_C}{\tau_A})(\sqrt{\tau_A} - \rho \sqrt{\tau_B})^2 \Rightarrow \\
\biggl(1-\rho^2\biggr)^2 \biggl(1 + \frac{\tau_C}{\tau_B}\biggr) \geq \biggl(\rho^2 + \frac{\tau_C}{\tau_A}\biggr)\biggl(\sqrt\frac{\tau_A}{\tau_B} - \rho\biggr)^2
\end{align}\n

\noindent as claimed.
\qed
\end{proof}

\section{Proof of Proposition \ref{brier-normal}}

We first prove a lemma for divergences of the quadratic scoring rule.
\begin{lemma}
\label{brier-diverge}
Let $s$ be the quadratic scoring rule, and let $\bm{p} \sim N(\mu, 1/\tau)$ and $\bm{\hat{p}} \sim N(\hat{\mu}, 1/\tau)$ be two normal distributions with the same accuracy $\tau$. Then
\s\begin{align}
S(\bm{\hat{p}},\bm{p}) - S(\bm{p},\bm{p}) = \frac{\tau}{\sqrt{2\pi}} [e^{-\frac{\tau}{4} (\mu - \hat{\mu})^2} - 1]
\end{align}\n
\end{lemma}

\begin{proof}
\s\begin{align*}
\bm{\hat{p}} \cdot \bm{p} &= \frac{\tau^2}{2\pi} \int\limits_{-\infty}^\infty \exp\Bigl[-\frac{\tau}{2}\Bigl\{(x - \mu)^2 + (x - \hat{\mu})^2\Bigr\}\Bigr] dx \\
&= \frac{\tau^2}{2\pi} \int\limits_{-\infty}^\infty  \exp\Bigl[-\frac{\tau}{2}\Bigl\{2 \Bigl(x - \frac{\mu + \hat{\mu}}{2}\Bigr)^2 + \frac{(\mu - \hat{\mu})^2}{2}\Bigr\}\Bigr] dx \\
&= \frac{\tau}{2 \sqrt{2\pi}} e^{-\frac{\tau}{4} (\mu - \hat{\mu})^2} \frac{2 \tau}{\sqrt{2\pi}} \int\limits_{-\infty}^\infty  \exp\Bigl[-\frac{2 \tau}{2}\Bigl(x - \frac{\mu + \hat{\mu}}{2}\Bigr)^2\Bigr] dx \\
&= \frac{\tau}{2 \sqrt{2\pi}} e^{-\frac{\tau}{4} (\mu - \hat{\mu})^2}
\end{align*}\n

For the quadratic scoring rule
\s\begin{align*}
S(\bm{\hat{p}},\bm{p}) - S(\bm{p},\bm{p}) = \E\limits_{r \sim \bm{p}} [S(\bm{\hat{p}},r) - S(\bm{p},r)] = 2 \bm{\hat{p}} \cdot \bm{p} - \bm{\hat{p}} \cdot \bm{\hat{p}} -  \bm{p} \cdot \bm{p}
\end{align*}\n

Therefore
\s\begin{align*}
S(\bm{\hat{p}},\bm{p}) - S(\bm{p},\bm{p}) = \frac{\tau}{2 \sqrt{2\pi}} [ 2 e^{-\frac{\tau}{4} (\mu - \hat{\mu})^2} - 1 - 1] =  \frac{\tau}{\sqrt{2\pi}} [e^{-\frac{\tau}{4} (\mu - \hat{\mu})^2} - 1]
\end{align*}\n
\qed
\end{proof}

{\em Local} truthfulness depends only on the sign of the second derivative of Alice's expectation of $\Pi_B(x; a,\hat{a},b)$, as shown in the following lemma.

\begin{lemma}
\label{locally_truthful}
Let $A_1 = a$. Assume that Alice, on her first prediction, announces her prediction as $\hat{a}$, where possibly $\hat{a} \neq a$. Alice will be locally truthful if
\s\begin{align}
\label{local_truth}
\frac{\partial^2}{\partial \hat{a}^2} \Bigl[ \Bigl\{\E\limits_{b \sim B_1|(A_1=a)} S(h(\hat{a},b),h(a,b))\Bigl\} - S(g(\hat{a}),g(a))\Bigr] \Bigr|_{\hat{a}=a} > 0
\end{align}\n
\end{lemma}

\begin{proof}
As $s$ is proper, for every $b$
\s\begin{align*}
&\frac{\partial}{\partial \hat{a}} S(h(\hat{a},b), h(a,b)) \Bigr|_{\hat{a}=a} = 0 &\frac{\partial}{\partial \hat{a}} S(g(\hat{a}), g(a)) \Bigr|_{\hat{a}=a} = 0 \\
&\frac{\partial^2}{\partial \hat{a}^2} S(h(\hat{a},b),h(a,b))  \Bigr|_{\hat{a}=a} < 0 &\frac{\partial^2}{\partial \hat{a}^2} S(g(\hat{a}),g(a))  \Bigr|_{\hat{a}=a} < 0
\end{align*}\n

Since $h(a,b)$ is a conditional distribution on $g(a)$, by the Law of Total Expectation we have
\s\begin{align*}
\E\limits_{b \sim B_1|(A_1=a)} \E \limits_{\lambda \sim h(a,b)} S(g(\hat{a}),\lambda) = \E\limits_{\lambda \sim g(a)} S(g(\hat{a}),\lambda) = S(g(\hat{a}),g(a))
\end{align*}\n

From the above and \eqref{Pi_B} we have that
\s\begin{align*}
 \frac{\partial}{\partial \hat{a}} \E\limits_{b \sim B_1|(A_1=a)}\E \limits_{\lambda \sim h(a,b)}\Pi_B(\lambda; a,\hat{a},b)  \Bigr|_{\hat{a}=a} = 0
\end{align*}\n

So $\hat{a}=a$ is an extremum point for $\E\limits_{b \sim B_1|(A_1=a)} \E \limits_{\lambda \sim h(a,b)}\Pi_B(\lambda; a,\hat{a},b)$, and so is a local minimum iff its second derivative is positive, or, equivalently, if \eqref{local_truth} holds. 
\qed
\end{proof}

The proof of the proposition follows.
 
\begin{proof}
We use Corollary \ref{min-truth}. Suppose that $A_1 = a$, and that Alice, when making her first prediction, pretends that $A_1 = \hat{a} := a + c'$, where $c' = \frac{\tau_A}{\tau_A + \tau_C}c$. I.e., by \eqref{acc_1}, Alice pretends her signal is $A_0 + c$.

By the Corollary, Alice is untruthful if \s$$\Delta(c;a, b) := \Bigl[S(h(\hat{a},b),h(a,b)) - S(h(a,b),h(a,b))\Bigr] - \Bigl[S(g(\hat{a}),g(a)) - S(g(a), g(a))\Bigr] < 0$$\n for every $a,  b$ and $c \neq 0$. $g(a)$, Alice's distribution at her first prediction, is given by \eqref{acc_1}. When $B_1 = b$, the inferred distribution from Alice's and Bob's predictions $h(a, b)$ is given by Proposition \ref{abc}. Substituting these, and using Lemma \ref{brier-diverge}:
\s\begin{align}
\label{delta}
\Delta(c; a, b) &= \frac{1}{\sqrt{2\pi}} \Bigl\{ \tau_{ABC}  [e^{-\frac{\tau_{ABC}}{4} \Bigl(c \frac{(\tau_A - \rho \sqrt{\tau_A \tau_B})}{{\tau_A - 2 \rho \sqrt{\tau_A \tau_B} + \tau_B + (1 - \rho^2) \tau_C}}\Bigr)^2} - 1] - \tau_{AC}  [e^{-\frac{\tau_{AC}}{4} \Bigl(c\frac{\tau_A}{\tau_A + \tau_C}\Bigr)^2} - 1]  \Bigr\} 
\end{align}\n

Letting $c \to \infty$ we deduce
\s\begin{align}
\lim\limits_{c \to \infty} \Delta(c; a, b) &= \frac{1}{\sqrt{2\pi}} (-\tau_{ABC} + \tau_{AC}) < 0
\end{align}\n

\noindent because $\tau_{ABC} > \tau_{AC}$, since $\bm{\pi_{ABC}}$'s distribution is conditional on $\bm{\pi_{AC}}$, and Bob's signal is informative. As this does not depend on $a$ nor on $b$, the setting is not truthful.

While not (globally) truthful, Alice may still be locally truthful. By Lemma \ref{locally_truthful}, we need to find out when $\frac{\partial}{\partial c^2} \Delta(c; a, b) \Bigr|_{c = 0} > 0$. From \eqref{delta}, noting that $\frac{\partial}{\partial x^2} y e^{-z x^2} \Bigr|_{x = 0} = - 2 y z$.
\s\begin{align}
\label{second-diff}
\frac{\partial}{\partial c^2} \Delta(c; a, b) \Bigr|_{c = 0} &= - 2 \tau_{ABC} \frac{\tau_{ABC}}{4} \Bigl(\frac{(\tau_A - \rho \sqrt{\tau_A \tau_B})}{{\tau_A - 2 \rho \sqrt{\tau_A \tau_B} + \tau_B + (1 - \rho^2) \tau_C}}\Bigr)^2 + 2  \tau_{AC} \frac{\tau_{AC}}{4} \Bigl(\frac{\tau_A}{\tau_A + \tau_C}\Bigr)^2 \\
& = - 2 \tau_{ABC} \frac{\tau_{ABC}}{4} \Bigl(\frac{(\tau_A - \rho \sqrt{\tau_A \tau_B})}{{\tau_{ABC} (1 - \rho^2)}}\Bigr)^2 + 2  \tau_{AC} \frac{\tau_{AC}}{4} \Bigl(\frac{\tau_A}{\tau_{AC}}\Bigr)^2 \\
&= \frac{1}{2} \tau_A^2 \Bigl[1 - \Bigl(\frac{(1 - \rho \sqrt{\frac{\tau_B}{\tau_A}})}{{(1 - \rho^2)}}\Bigr)^2\Bigr]
\end{align}\n

As this is true for every $b$, we deduce that local truthfulness holds for $1 - \rho^2 > 1 - \rho \sqrt{\frac{\tau_B}{\tau_A}}$. This entails $0 < \rho < \frac{\sigma_A}{\sigma_B}$, as claimed.
\qed
\end{proof}

\section{Proof of Proposition \ref{group}}

\begin{proof}
We show that truthfulness by all players is a Nash equilibrium, by showing that for each player, given truthfulness by every other player, if the player received her last signal, her best-response behavior is to make a truthful prediction. For suppose a player has received her last signal before the outcome is revealed. The player is aware of all her signals, as well as of previous predictions by her colleagues who have spoken. 

The player then hypothesizes the following: First, that she makes a truthful prediction now, and second, she hypothesizes some sequence $\mathcal{S}$ of predictions by other players, subject to the condition that all other players report truthfully their last signal (at least). The sequence ends in a last prediction $\bm{p}(\mathcal{S})$. Under the hypothesis, $\bm{p}(\mathcal{S})$ is the player's subjective posterior distribution of the outcome from knowing all information in the (hypothetical) setting. Furthermore, it is the distribution that will be scored by the mechanism. That is to say, the player's score expectation for the scenario is $S(\bm{p}, \bm{p})$. Now, if the player deviates, by predicting untruthfully, or by failing to predict, the last prediction, still under the hypothesis, will be $\bm{q}$. Now since $S$ is proper, $S(\bm{p}, \bm{p}) \geq S(\bm{q}, \bm{p})$ for any $\bm{q}$, and therefore there is no gain in deviating.

As this is true for {\em any} hypothetic future sequence $S$, predicting truthfully at present is a best-response regardless of the future.
\qed
\end{proof}

\section{Proof of Proposition \ref{lmsr-normal-discount}}

\begin{proof}
Referring to the proof of Proposition \ref{lmsr-normal}, from \eqref{crit}, clearly a discount factor $k(t)$ that satisfies
\s\begin{align*}
k (t_1) \frac{\tau_A}{\tau_A + \tau_C} - k(t_2) \frac{(\sqrt{\tau_A} - \rho \sqrt{\tau_B})^2}{(1 - \rho^2)(\sqrt{\tau_A} - \rho \sqrt{\tau_B})^2 + (1 - \rho^2)^2 (\tau_B + \tau_C)} \geq 0
 \end{align*}\n

will be effective. Rearranging this expression, the proposition follows.
\qed
\end{proof}

\section{Proof of Proposition \ref{brier-normal-discount}}

\begin{proof}
The logarithmic scoring rule is bounded. Referring to the proof of Propostion \ref{brier-normal}, from \eqref{delta} we see that \s$$\lim_{\hat{a} \to \infty} \frac{\E\limits_{b \sim B_1|(A_1=a)}\bigl[S(h(a,b),h(a,b) - S(h(\hat{a},b),h(a,b))\bigr]}{S(g(a),g(a)) - S(g(\hat{a}),g(a))} = \frac{\tau_{ABC}}{\tau_{AC}}$$\n is bounded. Furthermore, from \eqref{second-diff}, \s$$lim_{\hat{a} \to a} \frac{ \frac{\partial^2}{\partial \hat{a}^2} S(h(\hat{a},b),h(a,b))}{ \frac{\partial^2}{\partial \hat{a}^2} S(g(\hat{a}),g(a))} = \Bigl(\frac{(1 - \rho \sqrt{\frac{\tau_B}{\tau_A}})}{{(1 - \rho^2)}}\Bigr)^2$$\n which exists  whenever $|\rho| < 1$.
\qed
\end{proof}

\section{Proof of Theorem \ref{general}}

\begin{proof}
We prove the proposition by backward induction. Suppose that a discount factor $k(t)$ has been applied that would be sufficient to make all Alice-Bob-Alice interactions, with a ``composite'' Bob as described in the proposition, promptly truthful.

Assume that general prompt truthfulness is {\em not} in equilibrium. Call ``Alice'' the {\em last} player who is motivated to deviate from prompt truthfulness. This necessarily means that this is not Alice's last prediction, and that all players who make predictions between Alice's current prediction and her next one (call them {\em Bobbies}) are truthful.

Referring to Corollary \ref{min-truth} and its proof, it is readily seen that Alice maximizes her reward by minimizing the {\em aggregate} score of all Bobbies. Now, with a market scoring rule, in an interval where all players tell the truth, the aggregate score is the difference in score between final public information (public's distribution before Alice speaks again) and initial public information (public's distribution after Alice deviated). The sequence of events that led from initial information to final information has no effect on the aggregate score.

Any one of those Bobbies (call him Bob) has the same initial information, and, since all Bobbies speak the truth, the same final information. If we attribute all the signals of all Bobbies to Bob, we would again have the same initial information and same final information. By Corollary \ref{min-truth}, Alice would tell the truth to minimize Bob's expected aggregate score. She would therefore do the same to minimize the Bobbies' aggregate score, contradicting our assumption that she deviates. This proves that no player deviates from prompt truthfulness.
\qed
\end{proof}

\end{document}